%% file: l4dc2020-sample.tex
\documentclass[]{l4dc2021} 

\usepackage{bm,xpatch}
\usepackage[capitalize]{cleveref}
\usepackage{xcolor}
\usepackage{fp, tikz, pgfplots}
\usepackage{algorithm}
\usepackage{algorithmic}
\usepackage{mathtools}
\usepackage{multirow}
\usepackage{booktabs}
\usepackage{interval}
\usepackage{capt-of}
\usetikzlibrary{arrows}
\usetikzlibrary{shapes}
\usetikzlibrary{patterns}
\usetikzlibrary{decorations.pathmorphing}
\usetikzlibrary{decorations.pathreplacing}
\usetikzlibrary{decorations.shapes}
\usetikzlibrary{decorations.text}
\usetikzlibrary{shapes.misc}
\usetikzlibrary{decorations.markings}
\usetikzlibrary{arrows.meta}
\usetikzlibrary{backgrounds}
\usepgfplotslibrary{fillbetween}
\usepgfplotslibrary{statistics}

\newtheorem{assumption}{Assumption}

\DeclareMathOperator*{\argmin}{arg\,min}
\DeclareMathOperator*{\argmax}{arg\,max}

\newcommand{\Nmax}{{\bar{N}}}
\makeatletter
\xpatchcmd{\proof}{\topsep0\p@\@plus0\p@\relax}{}{}{}
\makeatother

\makeatletter
\def\blfootnote{\gdef\@thefnmark{}\@footnotetext}
\makeatother

\pgfplotsset{width=10\columnwidth /10, compat = 1.13, 
	height = 55\columnwidth /100, grid= major, 
	legend cell align = left, ticklabel style = {font=\scriptsize},
	every axis label/.append style={font=\small},
	legend style = {font=\tiny},title style={yshift=-7pt, font = \small} }

\title[The Impact of Data on the  Stability of Learning-Based Control]{The Impact of Data on the  Stability of Learning-Based Control -- Extended Version}
\usepackage{times}

\author{%
 \Name{Armin Lederer} \Email{armin.lederer@tum.de}\\
 \Name{Alexandre Capone} \Email{alexandre.capone@tum.de}\\
  \Name{Thomas Beckers} \Email{t.beckers@tum.de}\\
   \Name{Jonas Umlauft} \Email{jonas.umlauft@tum.de}\\
 \Name{Sandra Hirche} \Email{hirche@tum.de}\\
 \addr Chair of Information-oriented Control \\
 Department of Electrical and Computer Engineering\\
 Technical University of Munich \\
 D-80333 Munich, Germany\thanks{This paper is the extended version of \cite{Lederer2021b}. The official publication can be found at \url{http://proceedings.mlr.press/v144/lederer21a/lederer21a.pdf}.}
}

\begin{document}
	
\setlength{\textfloatsep}{3.0pt}
\setlength{\abovedisplayskip}{4.5pt}
\setlength{\belowdisplayskip}{4.5pt}
\setlength{\topsep}{4.0pt}

\maketitle

\begin{abstract}%
Despite the existence of formal guarantees for learning-based control approaches, the relationship between data and control performance is still poorly understood. In this paper, we propose a Lyapunov-based measure for quantifying the impact of data on the certifiable control performance. By modeling 
unknown system dynamics through Gaussian processes, we can determine the interrelation between model uncertainty and satisfaction of
stability conditions. This allows us to directly asses the impact of data on the provable stationary control performance, and thereby the value of the data
for the closed-loop system performance. Our approach is applicable to a wide variety of unknown nonlinear systems that are to be controlled by 
a generic learning-based control law, and the results obtained in numerical simulations indicate the efficacy of the proposed measure.

\end{abstract}


\section{Introduction}

Learning-based control is rapidly becoming an attractive alternative to traditional control approaches, particularly in settings with limited knowledge or prohibitive system complexity \citep{Deisenroth2015,Chua2018}. This has prompted a vast amount of research into the theoretical properties of control approaches based on nonparametric and probabilistic models obtained from supervised
machine learning,  yielding techniques that guarantee either safety or performance requirements \citep{Aswani2013, Berkenkamp2016, Beckers2019,Fisac2019,Capone2019,Gahlawat2020a,Lederer2020c}.

{While theoretical guarantees for control performance can be obtained in various settings, the direct relationship between collected data and control performance in learning-based control with nonparametric, probabilistic models is still poorly understood. In experimental design, the value of data is often quantified using information theoretical quantities, such as mutual information or entropy \citep{Pukelsheim2006}. Although these quantities have been used extensively to guide exploration and control strategies \citep{Hennig2012,Alpcan2015,Koller2018,Capone2020c}, they do not provide direct insight into the impact of data on the provable control performance. 
	In recent years, a handful of efforts has been carried out towards understanding the influence of collected data on control theoretic properties. \cite{Lederer2020a} have proposed a Lyapunov-based measure for quantifying the value of data points with respect to a specific control task based on Gaussian process priors. In a similar vein, \cite{Capone2021} have developed an algorithm to identify the most useful data points for successfully performing multiple control tasks. However, the technique presented in \cite{Lederer2020a} is only applicable to a restricted class of systems, and scalability is a challenge for the method of \cite{Capone2021}.
	
	In this work, we present a Lyapunov-based measure for quantifying the value of training data on the certifiable stationary performance of learning-based control of nonlinear systems. Based on Gaussian process models of unknown system dynamics, the model uncertainty is quantified and a necessary condition on the uncertainty 
	is derived for ensuring stability of unknown closed-loop systems. This condition is transformed into a required training data density, thereby 
	providing a measure for the impact of data on the provable stationary control performance. By considering a markedly richer class of systems,
	this work generalizes the results from \cite{Lederer2020a}.\looseness=-1
	
	The remainder of this paper is structured as follows. In \Cref{sec:ProbStat}, we formally state the considered problem. In \Cref{sec:GPR}, we briefly discuss how Gaussian processes are employed to model the unknown system, and present some preliminary results. Afterwards, in \Cref{sec:info_measure}, we derive the proposed information measure. We then discuss strategies for data selection, in \Cref{sec:data_selection}, after which we present some experimental results, in \Cref{sec:NumEval}. Finally, some concluding remarks are provided in \Cref{sec:Conc}.
	
	\section{Problem Statement}
	\label{sec:ProbStat}
	
	We consider a dynamical system
	\begin{align}
	\dot{\bm{x}}=\bm{g}(\bm{z})= \bm{A} \bm{f}(\bm{z}),
	\label{eq:dyn}
	\end{align}
	where $\bm{z}=\begin{bmatrix}
	\bm{x}^T&\bm{u}^T
	\end{bmatrix}^T\in\mathbb{R}^{d_z}$, $d_z=d_x+d_u$, is the concatenation of the 
	state $\bm{x}\in\mathbb{X}$ and the input 
	$\bm{u}\in\mathbb{U}$ for compact sets $\mathbb{X}\subset\mathbb{R}^{d_x}$ and $\mathbb{U}\subset\mathbb{R}^{d_u}$. We assume to
	know the matrix $\bm{A}\in\mathbb{R}^{d_x\times d_f}$, whereas the function
	$\bm{f}:\mathbb{R}^{d_x+d_u}\rightarrow \mathbb{R}^{d_f}$ is unknown. 
	This system 
	structure is very flexible
	and allows general multi-dimensional nonlinear systems ($\bm{A}=\bm{I}_{d_x}$, $d_f=d_x$), 
	as well as correlation in the outputs, 
	which can be found, e.g., in Euler-Lagrange systems due to the symmetry of the mass matrix \citep{Cheng2015}.
	We assume to have an approximate model
	$\hat{\bm{f}}:\mathbb{R}^{d_x+d_u}\rightarrow\mathbb{R}^{{d_f}}$ of the unknown function 
	$\bm{f}(\cdot,\cdot)$, which is often available in practice, as well as measurement data. This 
	yields the following formal assumption, which is discussed in detail in \citet{Umlauft2020}.\looseness=-1
	\begin{assumption}
		\label{as:traindat}
		A data set containing $N$ measurement pairs
		\begin{align}
		\mathbb{D}_N=\left\{ \bm{z}^{(n)}\coloneqq \begin{bmatrix}
		\bm{x}^{(n)}\\\bm{u}^{(n)}
		\end{bmatrix},\bm{y}^{(n)}=\bm{g}\left(\bm{z}^{(n)}\right)+\bm{\epsilon}^{(n)} \right\}_{n=1}^N
		\end{align}
		is available, where $\bm{\epsilon}^{(n)}\sim\mathcal{N}(0,\bm{\Sigma}_{\mathrm{on}})$ is i.i.d. Gaussian noise with 
		covariance matrix $\bm{\Sigma}_{\mathrm{on}}$.
	\end{assumption}
	Additionally, we assume that the unknown functions $f_i(\cdot)$ are well behaved, as 
	expressed in the following.
	\begin{assumption}
		\label{as:Lip}
		The unknown functions $f_i(\cdot)$ are Lipschitz continuous with Lipschitz constants $L_{f_i}$, 
		$i=1,\ldots,{d_f}$.
	\end{assumption}
	We assume to have a learning-based control law $\bm{\pi}:\mathbb{X}\times\left(\mathbb{X}\times\mathbb{U}\times\mathbb{R}^{d_x} \right)^N\times \mathbb{R}_{0,+}\rightarrow \mathbb{U}$
	that causally maps a state $\bm{x}$ to a control 
	input $\bm{u}$ depending on the previously observed training data $\mathbb{D}_N$ and the 
	current time $t\in\mathbb{R}_{0,+}$. The control law 
	is designed to achieve a control task, e.g., stabilization with respect to a reference point, or 
	tracking of a reference trajectory. The effectiveness of the control law with respect to this 
	task is measured via a 
	Lyapunov function\footnote{A Lyapunov function $V:\mathbb{X}\times\mathbb{R}_{0,+}\rightarrow\mathbb{R}_{0,+}$ is 
		positive definite, i.e., $V(\bm{x},t)\geq 0$ with equality if and only if $\bm{x}=0$.}
	$V:\mathbb{X}\times\mathbb{R}_{0,+}\rightarrow\mathbb{R}_{0,+}$ and its temporal derivative 
	$\dot{V}(\cdot,\cdot)$ along trajectories of the closed loop system defined through 
	\begin{align}
	\tilde{\bm{g}}(\bm{x})=\bm{g}\left( \begin{bmatrix}
	\bm{x}\\\bm{\pi}(\bm{x},\bm{x}^{(1)},\bm{u}^{(1)},\ldots,t)
	\end{bmatrix} \right),
	\label{eq:cl-sys}
	\end{align}
	which only depends on $\bm{x}$, as the control inputs are specified by the policy $\bm{\pi}(\cdot,\bm{x}^{(1)},\bm{u}^{(1)},\ldots,t)$.
	Since we do not know the function $\bm{f}(\cdot)$ but only have a training data set $\mathbb{D}_N$ 
	and an approximate model $\hat{\bm{f}}(\cdot)$, 
	the warranted control performance, measured through the size of the region in which $\dot{V}(\bm{x})\leq0$ is not guaranteed,
	strongly depends on the training data. In such a setting, it is crucial to understand the interrelation between training
	data and control performance, e.g., for determining where additional training data should be acquired to increase control performance, or when a subset of the training data must be selected to 
	reduce the computational complexity of learning. Therefore, we consider the problem of 
	determining the impact of data on stability certificates for learning-based control.
	
	\section{Gaussian Process Regression}
	\label{sec:GPR}
	For determining the impact of data on stability certificates, we employ a Gaussian process (GP) formulation, 
	such that we can relate the control performance to the model uncertainty. We first introduce 
	the foundations of Gaussian process regression in \cref{subsec:GP}, before we explain 
	the extension to multiple outputs, in \cref{subsec:MOGP}. 
	Finally, we propose an output-decoupling formulation using linear models of coregionalization and  
	propose a novel uniform error bound, in \cref{subsec:output decoup}.
	
	\subsection{Single Output Gaussian Processes}
	\label{subsec:GP}
	
	A Gaussian process $\mathcal{GP}(\hat{f}(\cdot,),k(\cdot,\cdot))$, uniquely defined through a 
	prior mean function $\hat{f}:\mathbb{R}^{d_z}\rightarrow\mathbb{R}$ and a covariance 
	function $k:\mathbb{R}^{d_z}\times\mathbb{R}^{d_z}\rightarrow\mathbb{R}_{0,+}$, is a 
	generalization of Gaussian distributions \citep{Rasmussen2006}. The prior mean function 
	$\hat{f}(\cdot)$ is often used to include approximate models in the regression, whereas the 
	covariance function describes prior assumptions on properties such as smoothness or 
	periodicity. A commonly used covariance function, on which we also focus in the following 
	analysis for clarity of exposition, 
	is the squared exponential kernel \looseness=-1 
	\begin{align}
	\label{eq:SE kernel}
	k(\bm{z},\bm{z}')=s_f^2\exp\left( -\frac{1}{2}(\bm{z}-\bm{z}')^T\bm{\Lambda}^{-1}(\bm{z}-\bm{z}')\right),
	\end{align}
	where $s_f^2\!\in\!\mathbb{R}_{0,+}$ and $\bm{\Lambda}\!\in\!\mathbb{R}^{d_z\times d_z}$, $\bm{\Lambda}\succ \bm{0}$, denote the signal variance and length
	scales, respectively. 
	
	In order to perform regression with the GP, we consider a scalar system \eqref{eq:dyn} with 
	${d_f}=1$. Then, the joint prior distribution of training targets $\bm{t}=\begin{bmatrix}
	y^{(1)}&\cdots&y^{(N)}
	\end{bmatrix}^T$ and the unknown function value $f(\bm{z})$ for input $\bm{z}$ 
	is given by
	\begin{align}
	\begin{bmatrix}
	\bm{t}\\f(\bm{z})
	\end{bmatrix}\sim\mathcal{N}\left( \begin{bmatrix}
	\hat{f}(\bm{Z})\\\hat{f}(\bm{z})
	\end{bmatrix},\begin{bmatrix}
	k(\bm{Z},\bm{Z})+\sigma_{\mathrm{on}}^2\bm{I}_N&k(\bm{Z},\bm{z})\\
	k^T(\bm{Z},\bm{z})&k(\bm{z},\bm{z})
	\end{bmatrix} \right),
	\end{align}
	where we use the abbreviations $\hat{f}(\bm{Z})\!\in\!\mathbb{R}^{N}$, $k(\bm{Z},\bm{Z})\!\in\!\mathbb{R}^{N\times N}$ and $k(\bm{Z},\bm{Z})\!\in\!\mathbb{R}^{N}$ 
	with elements defined as $\hat{f}_n(\bm{Z})\!=\!\hat{f}(\bm{z}^{(n)})$, $k_{n,n'}(\bm{Z},\bm{Z})\!=\!
	k(\bm{z}^{(n)},\bm{z}^{(n')})$ and $k_n(\bm{Z},\bm{z})\!=\!k(\bm{z}^{(n)},\bm{z})$, $n,n'\!=\!1,\ldots,N$, 
	respectively. By conditioning the GP on the training data, we obtain the posterior distribution\looseness=-1
	\begin{align}
	f(\bm{z})|y^{(1)},\ldots,y^{(N)},\bm{z}^{(1)},\ldots,\bm{z}^{(N)},\bm{z}\sim
	\mathcal{N}(\mu(\bm{z}),\sigma^2(\bm{z}))
	\label{eq:so GP dist}
	\end{align}
	with posterior mean and variance
	\begin{align}
	\mu(\bm{z})&=\hat{f}(\bm{z})+k^T(\bm{Z},\bm{z})\left(k(\bm{Z},\bm{Z})+\sigma_{\mathrm{on}}^2\bm{I}_N\right)^{-1}(\bm{t}-\hat{f}(\bm{Z}))\\
	\sigma^2(\bm{z})&=k(\bm{z},\bm{z})-k^T(\bm{Z},\bm{z})\left(k(\bm{Z},\bm{Z})+\sigma_{\mathrm{on}}^2\bm{I}_N\right)^{-1}k(\bm{Z},\bm{z}).
	\end{align}
	
	\subsection{Multiple-Output Gaussian Process Regression}
	\label{subsec:MOGP}
	
	In order to apply Gaussian processes to multiple-output regression problems, we can 
	proceed analogously to the single output case. For illustrative purposes, we assume for 
	now that $\bm{A}=\bm{I}_{d_x}$ and ${d_f}=d_x$ in \eqref{eq:dyn}, such that we have noisy 
	measurements of the functions $f_i(\cdot)$, $i=1,\ldots,{d_f}$ in the data set $\mathbb{D}_N$. 
	We start again with the prior GP distribution
	\begin{align}
	\label{eq:fprior}
	\bm{f}(\cdot)\sim\mathcal{GP}\left(\hat{\bm{f}}(\cdot),\bm{K}(\cdot,\cdot)\right),
	\end{align}
	where we have to consider a vector-valued prior mean function $\hat{\bm{f}}:\mathbb{R}^{d_z}\rightarrow\mathbb{R}^{d_f}$ 
	and a matrix kernel function $\bm{K}:\mathbb{R}^{d_z}\times\mathbb{R}^{d_z}\rightarrow\mathbb{R}^{{d_f}\times {d_f}}_{0,+}$, 
	in which each element $k_{m,m'}:\mathbb{R}^{d_z}\times\mathbb{R}^{d_z}\rightarrow\mathbb{R}_{0,+}$ 
	is a kernel. 
	By concatenating the training targets $\bm{y}_i^{(n)}$ in the vector 
	$\bm{t}^T=[
	y_1^{(1)}\ \cdots\ y_1^{(N)}\ y_2^{(1)}\ \cdots\ y_{d_x}^{(N)}]$
	and conditioning the joint distribution of $\bm{t}$ and $\bm{f}(\bm{z})$ on the training data 
	$\mathbb{D}_N$, analogously to \eqref{eq:so GP dist} we obtain a multivariate Gaussian distribution 
	with mean and covariance matrix
	\begin{align}
	\bm{\mu}(\bm{z})&=\hat{\bm{f}}(\bm{z})+\bm{K}^T(\bm{Z},\bm{z})  \left(\bm{K}(\bm{Z},\bm{Z})+\bm{\Sigma}_{\mathrm{on}}\otimes\bm{I}_N \right)^{-1}\left(\bm{t}-\hat{\bm{f}}(\bm{Z})\right)\\
	\bm{\Sigma}(\bm{z})&=\bm{K}(\bm{z},\bm{z})-\bm{K}^T(\bm{Z},\bm{z})  \left(\bm{K}(\bm{Z},\bm{Z})+\bm{\Sigma}_{\mathrm{on}}\otimes\bm{I}_N \right)^{-1}\bm{K}(\bm{Z},\bm{z}),
	\end{align}
	where we extend the shorthand notation from \cref{subsec:GP} using 
	\begin{align}
	\label{eq:Cmat}
	\bm{K}(\bm{Z},\bm{Z})&
	\begin{aligned}
	=\!\begin{bmatrix}
	\!k_{1,1}(\bm{Z},\bm{Z})\!\!&\!\!\cdots\!\!&\!\!k_{1,{d_f}}(\bm{Z},\bm{Z})\!\\
	\vdots\!\!&\!\!\ddots\!\!&\!\!\vdots\\
	\!k_{{d_f},1}(\bm{Z},\bm{Z})\!\!&\!\!\cdots\!\!&\!\!k_{{d_f},{d_f}}(\bm{Z},\bm{Z})\!
	\end{bmatrix}&&
	\bm{K}(\bm{Z},\bm{z})=\!\begin{bmatrix}
	\!k_{1,1}(\bm{Z},\bm{z})\!\!&\!\!\cdots\!\!&\!\!k_{1,{d_f}}(\bm{Z},\bm{z})\!\\
	\vdots\!\!&\!\!\ddots\!\!&\!\!\vdots\\
	\!k_{{d_f},1}(\bm{Z},\bm{z})\!\!&\!\!\cdots\!\!&\!\!k_{{d_f},{d_f}}(\bm{Z},\bm{z})\!
	\end{bmatrix}
	\end{aligned}\nonumber\\
	\hat{\bm{f}}(\bm{Z})&=\!\begin{bmatrix}
	\hat{f}_1(\bm{Z})&\cdots&\hat{f}_{d_f}(\bm{Z})
	\end{bmatrix}^T.
	\end{align}
	
	\subsection{Output Decoupling through Linear Models of Coregionalization}
	\label{subsec:output decoup}
	
	While various positive definite kernels are known for scalar regression, positive 
	definiteness is  
	a major challenge in the multiple-output approach presented in   
	\cref{subsec:MOGP} since it is not sufficient that each entry of $\bm{K}(\cdot,\cdot)$  
	is a covariance function. However, in the following we show that knowledge of the output 
	correlation structure in the form of a matrix $\bm{A}$ allows to define proper  
	kernel matrix functions via scalar covariance functions $k_i(\cdot,\cdot)$.  
	For this, we require the following assumption.\looseness=-1 
	\begin{assumption}
		\label{as:GPsample}
		Prior knowledge about the functions $f_i(\cdot)$ is expressed through 
		independent prior GP distributions with scalar kernels $k_i(\cdot,\cdot)$, $i=1,\ldots,{d_f}$, i.e., 
		\begin{align}
		\bm{f}(\cdot)\!\sim\!\mathcal{GP}\left(\hat{\bm{f}}(\cdot),\mathrm{diag}\left(\begin{bmatrix}
		k_1(\cdot,\cdot)&\cdots&k_{d_f}(\cdot,\cdot)
		\end{bmatrix}\right)\right).
		\end{align}
	\end{assumption}
	This assumption is not restrictive, since correlation in the training targets can be modeled through the 
	matrix $\bm{A}$, and it is frequently used in the case where $\bm{A}=\bm{I}_{d_x}$ holds \citep{Berkenkamp2015,Koller2018,Beckers2019, Hewing2020a}.
	
	Due to \cref{as:GPsample}, it directly follows that
	\begin{align}
	\bm{g}(\cdot)\sim\mathcal{GP}\left(\bm{A}\hat{\bm{f}}(\cdot),\bm{A}\mathrm{diag}\left(\begin{bmatrix}
	k_1(\cdot,\cdot)&\cdots&k_{d_f}(\cdot,\cdot)
	\end{bmatrix}\right)\bm{A}^T\right),
	\end{align}
	such that we can intuitively define a kernel matrix function through
	\begin{align}
	\bm{K}(\bm{z},\bm{z}')=\bm{A}\mathrm{diag}\left(\begin{bmatrix}
	k_1(\bm{z},\bm{z}')&\cdots&k_{d_f}(\bm{z},\bm{z}')
	\end{bmatrix}\right)\bm{A}^T.
	\end{align}
	It is trivial to show that this kernel parameterization is a special case of a linear model of 
	coregionalization \citep{Alvarez2011}, such that we can immediately extend the approach 
	in \cite{Duvenaud2014} to recover models for the individual functions $f_i(\cdot)$, as shown 
	in the following lemma\footnote{Proofs for all theoretical results can be found in the appendix.}.
	\begin{lemma}
		\label{lem1}
		Consider a nonlinear system \eqref{eq:dyn} with matrix $\bm{A}=\begin{bmatrix}
		\bm{a}_1&\cdots&\bm{a}_{d_f}
		\end{bmatrix}$ composed of column vectors $\bm{a}_i$, for which a training data set $\mathbb{D}_N$ 
		and prior distributions that satisfy 
		Assumptions~\ref{as:traindat} and \ref{as:GPsample}, respectively, are given. 
		Then, the posterior distributions are given by
		\begin{align}
		f_i(\bm{z})|\mathbb{D}\sim\mathcal{N}(\mu_i(\bm{z}),\sigma_i(\bm{z})),
		\end{align}
		where
		\begin{align}
		\label{eq:muf}
		\mu_i(\bm{z})&=\hat{f}_i(\bm{z})+\left(k_i^T(\bm{Z},\bm{z})\otimes\bm{a}_i^T\right)\left(\bm{K}(\bm{Z},\bm{Z})+\bm{\Sigma}_{\mathrm{on}}\otimes\bm{I}_N \right)^{-1}
		\left(\bm{t}-\hat{\bm{f}}(\bm{Z})\right)\\
		\sigma_i^2(\bm{z})&=k_i(\bm{z},\bm{z})
		-\left(k_i^T(\bm{Z},\bm{z})\otimes\bm{a}_i^T\right)\left(\bm{K}(\bm{Z},\bm{Z})+\bm{\Sigma}_{\mathrm{on}}\otimes\bm{I}_N \right)^{-1}\left( k_i(\bm{Z},\bm{z})\otimes\bm{a}_i\right).
		\label{eq:sigmaf}
		\end{align}
	\end{lemma}
	A crucial benefit of this decoupling of the outputs is that it allows the 
	application of scalar analysis methods to uniformly bound the regression error on the input domain $\mathbb{X}\times\mathbb{U}$
	as 
	proposed in \cite{Lederer2019}. This is formalized in the following theorem.
	\begin{theorem}
		\label{thm2}
		Consider a nonlinear system \eqref{eq:dyn}, a training data set $\mathbb{D}_N$, 
		and prior distributions satisfying Assumptions~\ref{as:traindat}-\ref{as:GPsample}, respectively. 
		For any $\delta\in(0,1)$, $\tau\in\mathbb{R}_+$, and $i=1,\ldots,{d_f}$, it holds that\looseness=-1
		\begin{align}
		\label{eq:unibound}
		P\left(|f_i(\bm{z})-\mu_i(\bm{z})|\leq \sqrt{\beta(\delta,\tau)}\sigma_i(\bm{z})+\gamma_i(\delta,\tau)\quad \forall \bm{z}\in\mathbb{X}\times\mathbb{U}\right)\geq 1-\delta,
		\end{align}
		where 
		\begin{align}
		\beta(\delta,\tau)=2d_x\log\left(1+\frac{r_0}{\tau}\right)-\log(\delta),
		\qquad \gamma_i(\delta,\tau)=(L_{\mu_i}+L_{f_i})\tau+\sqrt{\beta(\delta,\tau)L_{\sigma_i^2}\tau}.
		\end{align}
		Here, $L_{\mu_i}$ and $L_{\sigma_i^2}$ are the Lipschitz constants of the mean 
		and variance functions, respectively, and $r_0=\max_{\bm{z},\bm{z}'\in\mathbb{X}\times\mathbb{U}}\|\bm{z}-\bm{z}'\|$.
	\end{theorem}
	This theorem is a generalization of \citet[Lemma 1]{Lederer2020a} 
	and many properties directly transfer. Small error bounds can be achieved through small 
	GP standard deviations $\sigma_i(\bm{z})$, which corresponds to high data densities.
	This resembles well-known relationships from scattered data approximation \citep{Wendland2005} 
	and Bayesian optimization \citep{Srinivas2012}. 
	The dependence of the uniform error bound 
	\eqref{eq:unibound} on the constants $\gamma_i(\delta,\tau)$ does not affect this behavior, since
	they can be chosen arbitrarily small, and convergence to $0$ can be shown under weak 
	assumptions on $\sigma_i(\bm{z})$ \citep{Lederer2019}.
	In general, the constant
	$\tau$ trades-off the effect of the data independent terms $\gamma_i(\delta,\tau)$ 
	and the posterior standard deviations $\sigma_i(\bm{z})$ on the error bound. Therefore, $\tau$ should be chosen such that  the uncertainty dependence of the bound dominates, i.e., $\sqrt{\beta(\delta,\tau)}\sigma_i(\bm{z})\!\gg\! \gamma_i(\delta,\tau)$. \looseness=-1
	\begin{remark}
		\label{rem1}
		\cref{thm2} admits the counterintuitive behavior that adding training samples can lead 
		to a locally higher uniform error bound. This is due to the fact that 
		adding data in some regions can increase the Lipschitz constants of 
		$\mu(\cdot)$ and $\sigma(\cdot)$ \citep{Lederer2019}, and thereby increase the uniform error bound in other 
		regions. Note that a similar argument holds for uniform
		error bounds based on RKHS theory \citep{Srinivas2012, Chowdhury2017a}.
	\end{remark}

	\section{Control-Based Information Measures}
	\label{sec:info_measure}
	
	While the uniform error bound in \cref{thm2} establishes a connection between the training 
	data distribution, represented by the posterior GP variance, and the regression performance,
	it is ignorant of the control task. In order to measure the importance of data for control
	performance, we consider the Lyapunov stability conditions \citep{Khalil2002} for the closed 
	loop system, which require a negative derivative of the Lyapunov function $V(\cdot,\cdot)$,
	i.e., 
	\begin{align}
	\dot{V}(\bm{x},t)=(\nabla_{\bm{x}} V(\bm{x},t))^T\bm{A}\tilde{\bm{f}}(\bm{x})+\frac{\partial}{\partial t}V(\bm{x},t),
	\end{align}
	where we employ the shorthand notation $\tilde{\bm{f}}(\bm{x})=\bm{f}\big( \begin{bmatrix} \bm{x}^T&\bm{\pi}^T(\bm{x},\bm{x}^{(1)},\bm{u}^{(1)},\ldots,t)	\end{bmatrix}^T \big) $, which is used analogously for the GP mean 
	$\tilde{\bm{\mu}}(\cdot)$ and variance $\tilde{\bm{\sigma}}^2(\cdot)$. 
	Although the function $\tilde{\bm{f}}(\cdot)$ is unknown, we can bound the Lyapunov function derivative 
	based on the uniform error bound \eqref{eq:unibound}, which yields
	\begin{align}
	\label{eq:pess_stab}
	\dot{V}(\bm{x},t)&
	\leq \dot{V}_{\mathrm{nom}}(\bm{x},t)+\dot{V}_{\bm{\sigma}}(\bm{x},t),
	\end{align}
	where we decouple the bound into the nominal component based on the GP mean 
	\begin{align}
	\label{eq:nom Lyap}
	\dot{V}_{\mathrm{nom}}(\bm{x},t)&=\!(\nabla V(\bm{x},t))^T\bm{A}\tilde{\bm{\mu}}(\bm{x})\!+\!\frac{\partial}{\partial t}V(\bm{x},t),
	\end{align}
	and an uncertain component depending on the GP standard deviation
	\begin{align}\label{eq:uncert Lyap}
	\dot{V}_{\bm{\sigma}}(\bm{x},t)&=\!\begin{bmatrix}
	\left|\left(\nabla_{\bm{x}} V(\bm{x},t)\right)^T\!\bm{a}_1\right|\!&\!\cdots\!&\!\left|\left(\nabla_{\bm{x}} V(\bm{x},t)\right)^T\!\bm{a}_{d_f}\right|
	\end{bmatrix}\left(\sqrt{\beta(\delta,\tau)}\tilde{\bm{\sigma}}(\bm{x})\!+\!\bm{\gamma}(\delta,\tau)\right),
	\end{align}
	with $\bm{\gamma}(\delta,\tau)\!=\!\begin{bmatrix}
	\gamma_1(\delta,\tau)&\cdots&\gamma_{d_f}(\delta,\tau)
	\end{bmatrix}^T$. Since the uncertain derivative component 
	$\dot{V}_{\bm{\sigma}}(\cdot,\cdot)$ is non-negative, a positive nominal derivative component $\dot{V}_{\mathrm{nom}}(\cdot,\cdot)$
	directly implies that \eqref{eq:pess_stab} violates the Lyapunov stability conditions regardless of the GP posterior variance. 
	Therefore, we assume  
	$\dot{V}_{\mathrm{nom}}(\bm{x},t)\!<\! 0$ in the following,  essentially requiring that the control law can stabilize the dynamical system defined by $\bm{\mu}(\cdot)$.
	As a result, stability of the closed-loop system depends on the magnitude of the uncertain Lyapunov 
	function derivative $\dot{V}_{\bm{\sigma}}(\cdot,\cdot)$, which is strongly influenced by posterior 
	GP standard deviations $\tilde{\bm{\sigma}}(\cdot)$. Although this establishes a direct relationship 
	between the training data density and the control task, the dependency of $\tilde{\bm{\sigma}}(\cdot)$
	on training samples is highly nonlinear and the computation of $\tilde{\bm{\sigma}}(\cdot)$ is
	computationally expensive. In order to mitigate these issues, we introduce the weighted $M$-fill distances,
	in analogy to the $M$-fill distance proposed in \cite{Lederer2020a}.\looseness=-1
	\begin{definition}
		\label{def:filldist}
		The weighted $M$-fill distance $\phi_i(\bm{x},\mathbb{D}_N)$ for function $f_i(\cdot)$, $i=1\ldots,d_f$, at a point $\bm{x}$ is defined 
		as the minimum radius $\varphi$ of a ball with center $\tilde{\bm{z}}=[\bm{x}^T\ \bm{\pi}^T(\bm{x})]^T$, such that the ball contains $M$ 
		samples $\bm{z}^{(n)}$, i.e., 
		\begin{subequations}
			\begin{align}
			\!\tilde{\phi}_i(\bm{x},\mathbb{D}_N)=&\min\limits_{\phi\in\mathbb{R}_{+,0}}\varphi\\ 
			\text{such that} &\left| \left\{ \bm{z}^{(n)} \in \mathbb{D}_N: \left(\tilde{\bm{z}}-\bm{z}^{(n)}\right)^T\bm{\Lambda}_i^{-1}\left(\tilde{\bm{z}}-\bm{z}^{(n)}\right) \leq \varphi^2 \right\} \right|\geq M,\!
			\end{align}
		\end{subequations}
		where $|\cdot|$ denotes the cardinality of the set and we use the abbreviation $\bm{\pi}(\bm{x})\!=\!\bm{\pi}(\bm{x},\bm{x}^{(1)}\!,\bm{u}^{(1)}\!,\ldots,t)$.
	\end{definition}
	The weighted $M$-fill distances measure the distance from a test point $\tilde{\bm{z}}\!=\![\bm{x}^T\ \bm{\pi}^T(\bm{x})]^T$ to the $M$ closest  
	training samples in the Mahalonobis distance metric induced by the length scales $\bm{\Lambda}_i$  
	of the squared exponential kernels \eqref{eq:SE kernel}. 
	By choosing a small number $M\!\ll\! N$, only training points in the proximity of the test
	point $\tilde{\bm{z}}=[\bm{x}^T\ \bm{\pi}^T(\bm{x})]^T$ are relevant  
	for the weighted $M$-fill distance $\tilde{\phi}_i(\bm{x},\mathbb{D}_N)$. This allows us to measure the local 
	data density in a flexible way, where high training data densities are indicated by low values 
	of $\tilde{\phi}_i(\bm{x},\mathbb{D}_N)$. Moreover, it is possible to bound the posterior GP variances $\tilde{\bm{\sigma}}^2(\bm{x})$ in terms of the weighted $M$-fill distances $\tilde{\phi}_i(\bm{x},\mathbb{D}_N)$\footnote{
		A bound for the posterior variance in terms of the $M$-fill distance is derived in the appendix.}.
	We
	exploit this property in the following theorem to derive conditions that guarantee that the summands of the 
	uncertain Lyapunov derivative $\dot{V}_{\bm{\sigma}}(\cdot,\cdot)$ are upper bounded by functions 
	$\xi_{i}:\mathbb{R}^{d_z}\times\mathbb{R}_{0,+}\rightarrow\mathbb{R}_{0,+}$, $i=1,\ldots,d_f$. For suitably chosen functions
	$\xi_{i}(\cdot,\cdot)$, the satisfaction of these bounds implies stability of the closed-loop system.\looseness=-1
	\begin{theorem}
		\label{thm6}
		Choose $\tau$ such that $\sqrt{\beta(\delta,\tau)}\tilde{\sigma}_i(\bm{x})>\gamma_i(\delta,\tau)$ holds for all
		$\bm{x}\in\mathbb{X}$ and $\xi_i: \mathbb{R}^{d_z}\times\mathbb{R}_{0,+} \rightarrow \mathbb{R}_{0,+}$ such that\looseness=-1
		\begin{align}
		\label{eq:xi_cond}
		\dot{V}_{\sigma_{i,0}}^2(\bm{x},t)=4\beta(\delta,\tau)s_{f_i}^2 \left|\left(\nabla_{\bm{x}} V(\bm{x},t)\right)^T\!\bm{a}_{i}\right|^2>\xi_{i}^2(\bm{x},t).
		\end{align}
		If the $M$-fill distance $\tilde{\phi}_i(\cdot,\mathbb{D}_N)$ satisfies $\tilde{\phi}_{i}^2(\bm{x},\mathbb{D}_N)\leq \bar{\phi}_{i}^2(\bm{x},t)+\theta_i^2$ 
		for all $\bm{x}\in\mathbb{X}$, where
		\begin{align}
		\label{eq:phibar}
		\bar{\phi}_{i}^2(\bm{x},t)&=-\log\left( 1-\frac{\xi_{i}^2(\bm{x},t)}{\dot{V}_{\sigma_{i,0}}^2(\bm{x},t)} \right)\\
		\theta_i^2&=\log\left(s_{f_i}^2\|\bm{a}_i\|_2^2\right)-\log\left(\max\limits_{m=1,\ldots,d_x}\sum\limits_{n=1}^{{d_f}} |a_{m,n}| \|\bm{a}_n\|_1s_{f_n}^2+\frac{\lambda_{\max}(\Sigma_{\mathrm{on}})}{M} \right),
		\label{eq:theta}
		\end{align}
		then, with probability of at least $1-\delta$, it holds for all $\bm{x}\in\mathbb{X}$ that
		\begin{align}\label{eq:Vdot_sigmai}
		\dot{V}_{\sigma_i}(\bm{x},t)=\left(\sqrt{\beta(\delta,\tau)}\tilde{\sigma}_i(\bm{x})+\gamma_i(\delta,\tau)\right)\left|\left(\nabla_{\bm{x}} V(\bm{x},t)\right)^T\!\bm{a}_{i}\right|\leq \xi_{i}(\bm{x},t).
		\end{align}	
	\end{theorem}
	Condition \eqref{eq:xi_cond} is necessary to ensure the existence of the logarithm in \eqref{eq:phibar}, 
	but it is not restrictive since $\xi_{i}^2(\cdot,\cdot)$ are upper bounds. Hence, we can simply tighten 
	the bounds until $\xi_{i}^2(\cdot,\cdot)$ satisfies condition \eqref{eq:xi_cond}. The 
	expressions \eqref{eq:phibar} and 
	\eqref{eq:theta} have different roles. The values $\theta_i$ express the difficulty of recovering 
	the functions $f_i(\cdot)$ from the noisy measurements of $\bm{A}\bm{f}(\cdot)$, which in turn depends on the signal 
	variances $s_{f_i}^2$ of the independent covariance functions $k_i(\cdot,\cdot)$ and the magnitude of the
	elements of $\bm{A}$.  In contrast, $\bar{\phi}_{i}^2(\cdot,\cdot)$ captures the dependency on the 
	control task. The numerator in \eqref{eq:phibar} corresponds to a summand of the uncertain 
	Lyapunov derivative component \eqref{eq:uncert Lyap} without any training data as $s_{f_i}^2$ corresponds to the prior GP variance. Therefore, 
	$\bar{\phi}_{i}^2(\bm{x},t)$ goes to $\infty$ when the prior uncertain Lyapunov component converges
	to the bound $\xi_{i}^2(\bm{x},t)$. This intuitively reflects the fact that no data are required when the GP prior is already sufficient to guarantee \eqref{eq:Vdot_sigmai}.
	
	As mentioned previously, \cref{thm6} can be used to analyze the stability of the closed loop system.
	More specifically, stability is guaranteed if the negated nominal Lyapunov derivative $\dot{V}_{\mathrm{nom}}(\cdot,\cdot)$ is 
	larger than the uncertain component $\dot{V}_{\bm{\sigma}}(\cdot,\cdot)$. This can be trivially checked
	with \cref{thm6} by defining $\xi_{j}(\cdot,\cdot)$ such that 
	$\sum_{i=1}^{d_f}\xi_{i}(\bm{x},t)\leq |\dot{V}_{\mathrm{nom}}(\bm{x},t)|$. A natural 
	choice satisfying this condition together with constraint \eqref{eq:xi_cond} is given by
	\begin{align}
	\label{eq:xi}
	\xi_{i}(\bm{x},t)=\min\left\{
	-\frac{\dot{V}_{\mathrm{nom}}(\bm{x},t)\|\bm{a}_i\|_1}{\sum_{j=1}^{{d_f}}\|\bm{a}_j\|_1},
	\dot{V}_{\sigma_{i,0}}(\bm{x},t)- \nu
	\right\},
	\end{align}
	where $\nu\in\mathbb{R}_+$ is an arbitrarily small constant. Furthermore, convergence rates can be examined in a similar way. For example, an exponential rate of convergence is achieved by guaranteeing that a requirement similar to \eqref{eq:xi} is satisfied, where $\dot{V}_{\mathrm{nom}}(\bm{x},t)$ is replaced by $\dot{V}_{\mathrm{nom}}(\bm{x},t) + {V}(\bm{x},t)$.
	
	Due to the intuitive interpretation of \cref{thm6}, we propose to use it as the basis for a 
	measure of the importance of training data for control. This naturally leads to the definition 
	of the $\rho$-gap. 
	\begin{definition}
		\label{def:rho_gap}
		The $\rho$-gap is defined as 
		\begin{align}
		\rho(\bm{x},t,\mathbb{D}_N)=\sum\limits_{i=1}^{d_f}\max\{0,\phi_{i}^2(\bm{x},\mathbb{D}_N)-\bar{\phi}_{i}^2(\bm{x},t)-\theta_i^2\}. 
		\end{align}
	\end{definition}
	Essentially, the $\rho$-gap measures the discrepancy between the required data density, 
	which is expressed through $\bar{\phi}_{i}^2(\bm{x},t)+\theta_i^2$ and depends on the 
	desired bounds $\xi_i(\cdot,\cdot)$, the Lyapunov derivative and the signal standard 
	deviations $s_{f_i}^2$, and the actual data density represented by the $M$-fill 
	distances $\phi_{i}^2(\bm{x},\mathbb{D}_N)$, which are independent of the control
	problem and only depend on the available data.\looseness=-1
	
	\section{Data Selection Strategies}
	\label{sec:data_selection}
	Based on the information measure proposed in 
	Section~\ref{sec:info_measure}, the data set can be preprocessed to contain only 
	the most relevant information for the given control task. This becomes 
	particularly important in scenarios where the prediction of the GP model must 
	be performed under tight real-time constraints. Even with a precomputation of 
	the matrix inverse in~\eqref{eq:muf} (which takes $\mathcal{O}(N^3)$), the~$N$ 
	kernel evaluations (for $k_m^T(\bm{Z},\bm{z})$) and the 
	corresponding multiplications ($\mathcal{O}(N)$ for the posterior mean, 
	$\mathcal{O}(N^2)$ for the posterior variance) must still be performed online. 
	In practice, this imposes an upper bound for the number of points that can be considered by the model.
	We formulate the resulting computational constraint independent of the hardware 
	and specific real-time limit as follows.\looseness=-1
	\begin{assumption}
		\label{as:maxN}
		The computational constraints allow a maximum of~$\Nmax$ data points to be 
		considered by the GP regression model~\eqref{eq:muf}.
	\end{assumption}
	For the non-trivial case $N>\Nmax$, this makes a selection of an active data 
	set~$\mathbb{D}_\Nmax\subset\mathbb{D}_N$ necessary.
	Such a data selection has been considered for general function learning 
	in~\citet{Krause2008}, and specifically for control tasks 
	in~\citet{Umlauft2020b}. But both employ entropy-based criteria, which 
	only aims to optimize the precision of the model but does not consider the 
	closed-loop control performance. Therefore, we utilize the 
	$\rho$-gap as the measure for the control performance of a data 
	set to find the optimal active data set\footnote{To simplify 
		notation, we 
		introduce the index set~$\mathbb{I} = \{1,\ldots,N\}$ and all possible 
		subsets 
		with size~$\Nmax$, denoted as~$\mathbb{P}^{\mathbb{I}}_{\Nmax}$. 
		$\mathbb{i}(i)$ denotes the~$i$-th entry of the index set~$\mathbb{i}$.} 
	\begin{align}
	\label{eq:opt_activeData}
	\mathbb{i}^* = \argmin\limits_{\mathbb{i} \in \mathbb{P}^{\mathbb{I}}_{\Nmax}}\max\limits_{ 
		\begin{subarray}{l}
		\bm{t} \in \mathbb{T}, 
		\bm{x}\in\mathbb{X} 
		\end{subarray}} 
	\rho\left(\bm{x},t,\mathbb{i}\right), 
	\end{align} 
	where $\mathbb{T} = \rinterval{0}{T}$ with initial time~$t_0 \in \mathbb{R}_{0,+}, 
	t_0 < \infty$ and (possibly infinite) final time~$T\in \mathbb{R}\cup \infty $. The 
	optimal active data set is then given by~$\mathbb{D}_\Nmax = 
	\left\{\bm{z}^{(\mathbb{i}(i))},\bm{y}^{(\mathbb{i}(i))}\right\}_{i=1}^\Nmax$.
	If the 
	desired trajectory has a wide spread or~$\Nmax$ is small, then the selected 
	subset might not lead to a satisfactory control performance. For such a case, 
	we can partition the task in $S\in \mathbb{N}$ time intervals 
	$\mathbb{T}_0 = 
	\rinterval{0}{t_1}$,~$\mathbb{T}_1 = \rinterval{t_1}{t_2}$,~$\ldots$, 
	$\mathbb{T}_S = \rinterval{t_S}{T}$
	and compute the corresponding optimal subsets 
	$\mathbb{D}_\Nmax^{\mathbb{T}_0}$, $\ldots$, $\mathbb{D}_\Nmax^{\mathbb{T}_S}$.\looseness=-1
	
	Due to its mixed nature (combinatorial in $\mathbb{i}$, continuous in~$t$ and~$\bm{x}$), 
	the optimization problem~\eqref{eq:opt_activeData} is not trivial to solve. 
	However, the optimization can be performed offline, assuming that 
	sufficient memory capacity 
	is available to store all precomputed subsets.
	Furthermore, in the field of function learning it has been shown that greedy 
	algorithms can show near-optimal behavior~\citep{Krause2008}.
	Therefore, we propose the greedy data selection procedure shown in 
	Algorithm~\ref{alg:data_selection}.
	\begin{algorithm2e}[t]
		\caption{Greedy optimization for optimal subset selection}
		\label{alg:data_selection}
		\SetAlgoLined
		\DontPrintSemicolon
		\KwIn{$\mathbb{D}_N$, 
			$\rho(\cdot,\cdot)$,~$\mathbb{T}_0,\ldots,\mathbb{T}_S$}
		\KwOut{$\mathbb{D}_\Nmax^{\mathbb{T}_0}$,$\ldots$,$\mathbb{D}_\Nmax^{\mathbb{T}_S}$}
		\For{$\mathbb{T}=\mathbb{T}_0,\ldots,\mathbb{T}_S$}{
			$\mathbb{D}_\Nmax^{\mathbb{T}} \gets \emptyset$, $\mathbb{I} = 
			\{1,\ldots, N\}$ \;
			\For{$n=0,\ldots,\Nmax$}{
				$i^*, t^* \gets\argmax\limits_{
					\begin{subarray}{l} 
					i\in \mathbb{I},
					t\in \mathbb{T}
					\end{subarray}}  
				\rho\left(\bm{x}^{(i)},t,\mathbb{D}_{\bar{N}}^{\mathbb{T}}\right)$ \; 
				$\mathbb{D}_\Nmax^{\mathbb{T}} \gets \mathbb{D}_\Nmax^{\mathbb{T}} 
				\cup \left\{\bm{z}^{(i^*)}, \bm{y}^{(i^*)}\right\}$ \;
				$\mathbb{I} \gets \mathbb {I} \setminus \{i^*\}$ \;
			}
			
		}
	\end{algorithm2e}

	\section{Numerical Evaluation}
	\label{sec:NumEval}

	In order to evaluate the proposed importance measure, 
	we consider the nonlinear system
	\begin{align}
	\dot{\bm{x}}=\bm{x}+\frac{1}{1+\exp(-2x_1)}\begin{bmatrix}
	1\\-1
	\end{bmatrix}+0.5\begin{bmatrix}
	\sin(\pi x_2)\\\cos(\pi x_1)
	\end{bmatrix}+\bm{u},
	\label{eq:sys_ex}
	\end{align}
	which is a slight modification of the example proposed in \citep{Umlauft2018}. We assume a prior model $\hat{\bm{f}}(\bm{z})=\bm{x}+\bm{u}$, and define the kernel matrix using 
	\begin{align}
	\bm{A}=\begin{bmatrix}
	1&0\\-1&1
	\end{bmatrix}
	\end{align}
	and squared exponential kernels $k_1(\bm{x},\bm{x}')$, $k_2(x_1,x_1')$. This ensures that the correlation
	between the outputs caused by the second summand in \eqref{eq:sys_ex} is properly modeled. We employ a 
	control law 
	\[
	\bm{\pi}(\bm{x},t)=-(\bm{\mu}(\bm{x})+K(\bm{x}-\bm{x}_{\mathrm{ref}}(t))-\dot{\bm{x}}_{\mathrm{ref}}(t)),
	\]
	with gain $K\!=\!15$ and references $\bm{x}_{\mathrm{ref}}(t)\!=\![c_1\sin(t)\ c_2\cos(t)]^T$ with  
	randomly drawn $c_i\!\sim\!\mathcal{N}(0,1)$.  
	Nominal stability of the closed loop is shown using the Lyapunov function 
	$V(\bm{x},t)\!=\!(\bm{x}\!-\!\bm{x}_{\mathrm{ref}}(t))^T(\bm{x}\!-\!\bm{x}_{\mathrm{ref}}(t)).$
	The training set is generated by simulating the closed-loop system with prior mean $\bm{\mu}(\cdot)\!=\!\hat{\bm{f}}(\cdot)$ and sampling $N\!=\!100$ data points during the interval $t\!\in\![0,T]$ with
	$T\!=\!10$. We divide the period of the reference trajectory into $S\!=\!10$ equally long 
	intervals $\mathbb{T}_s$ and select subsets of cardinality $\bar{N}\!=\!10$ for each interval using  
	\cref{alg:data_selection} with $M\!=\!1$ and $\xi_i(\cdot)$ as defined in \eqref{eq:xi}. \looseness=-1

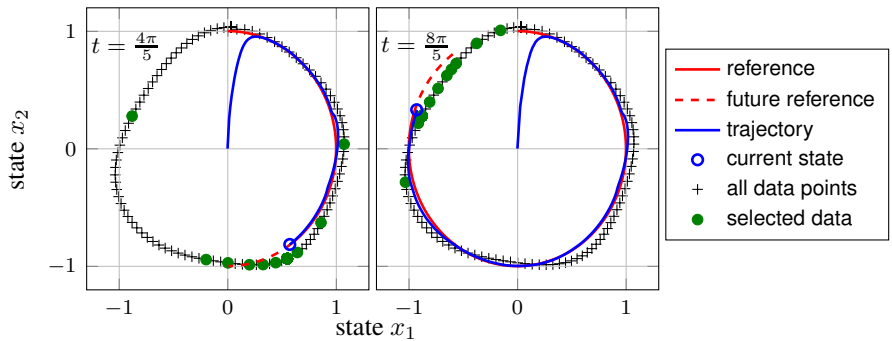
\begin{figure}[t]
	\begin{center}
		\vspace{0.15cm} 
		\input{figure/time_var.tex} 
		\vspace{-0.6cm}\caption{Snapshots of the desired and actual trajectory of the dynamical system. Data selected by \cref{alg:data_selection} lies 
			close to the future reference  and ensures a high tracking accuracy.}\vspace{-0.5cm} 
		\label{fig:time_var} 
		\vspace{-0.1cm} 
	\end{center}   
\end{figure}   
	
\begin{table}[b] 
	\vspace{-0.2cm} 
	\small
	\captionof{table}{Tracking errors and GP prediction times  
		resulting from different subset selection criteria. The $\rho$-gap significantly outperforms  
		existing methods regarding control performance.} 
	\label{tab:sum} 
	\centering
	\vspace{-0.3cm} 
	\begin{tabular}{lcccc} 
		\toprule
		criterion & full data set& mutual information & mutual information& $\rho$-gap\\ 
		&           & w.r.t. uniform grid& w.r.t. reference  &         \\ \midrule 
		steady-state MSE ($\cdot 10^{-3}$)& 1.15 & 1.32 & 0.38 & \textbf{0.16} \\ 
		prediction time ($\mu s$)         & 437  & \textbf{45.0} & \textbf{45.0} & \textbf{45.0} \\ \bottomrule   
	\end{tabular}
	\vspace{0.0cm} 
		
\end{table}

Snapshots of the selected subsets and the resulting system trajectories are illustrated in \cref{fig:time_var}.
It can be clearly seen that the training samples are chosen close to the reference in the considered
time intervals $\mathbb{T}_s$. This is because the feedback $K(\bm{x}\!-\!\bm{x}_{\mathrm{ref}})$
ensures that the Lyapunov stability condition is satisfied far away from the reference regardless of training data. Moreover, the data density grows as the distance to the reference decreases due 
to the vanishing effect of the feedback in its proximity. \looseness=-1

We evaluate our technique by carrying out $100$ control law roll-outs with randomly drawn  
trajectory parameters $c_i$. Furthermore, we compare the results with the performance of the full data set, as well as data selected with a greedy maximization of the mutual information with respect to a uniform grid over $[-1.5,1.5]^2$, and maximization of the mutual information with respect to the considered trajectory interval \citep{Umlauft2020b}. The results are depicted in \cref{tab:sum}. A reduction in the average computation time  
by a factor of approximately $10$ for all subset selection methods can observed, which is a straightforward  
consequence of the linear complexity  
of predictions \citep{Rasmussen2006}. Moreover, the steady state mean squared tracking error  
is the lowest for the subset selected based on the $\rho$-gap, and even smaller than for the full 
data set. The approach from  
\citep{Umlauft2020b} exhibits a similar performance. Although it might seem unintuitive that reducing the number of samples benefits control performance, this is not excluded by \cref{thm2}, as discussed in Remark~\ref{rem1}. This underlines the  
importance of selecting training data for learning-based control.  \looseness=-1  

\section{Conclusion}
\label{sec:Conc}
We presented the $\rho$-gap, a measure that quantifies the value of data for a broad class of control tasks. The proposed quantity is used to identify the optimal data set for control tasks under computational constraints. Simulations demonstrate that the data subsets selected using the presented measure
are highly correlated with the control task and can even be beneficial for control performance.\looseness=-1  
\acks{This work was supported by the European Research Council Consolidator Grant “Safe data-driven control for human-centric systems (CO-MAN)” under  grant  agreement  number  864686.  Armin  Lederer  gratefully  acknowledges financial support from the German Academic Scholarship Foundation.\looseness=-1}

\bibliography{myBib}

\appendix

\section{Component-wise Uniform Error Bounds for Multiple-Output Gaussian Process Regression}

\par\noindent{\bfseries\upshape Proof of Lemma~\ref{lem1}\ }
It can be easily checked that we can express the kernel matrix function as
\begin{align}
\bm{K}(\bm{z},\bm{z}')=\sum\limits_{i=1}^{d_f}\bm{B}_i k_i(\bm{z},\bm{z}'),
\label{eq:k sum}
\end{align}
where 
\begin{align}
\bm{B}_i=\begin{bmatrix}
a_{1,i}\\\vdots\\a_{d_x,i}
\end{bmatrix}\underbrace{\begin{bmatrix}
	a_{1,i}&\cdots&a_{d_x,i}
	\end{bmatrix}}_{\bm{a}_i^T}.
\end{align}
Therefore, the kernel matrix is linear in the scalar kernel functions $k_m(\cdot,\cdot)$, 
such that the posterior of $\bm{a}_i f_i(\cdot)$ can be obtained as
\begin{align}
	\bm{a}_i f_i(\cdot)|\mathbb{D}_N\sim \mathcal{N}(\mu_{\bm{a}_i f_i}(\cdot),\bm{\Sigma}_{\bm{a}_if_i}(\cdot)),
\end{align}
where
\begin{align}
	\mu_{\bm{a}_i f_i}(\bm{z})&=\left(\bm{B}_i\otimes k_i^T(\bm{Z},\bm{z})  \right)\left( 
	\bm{\Sigma}_{\mathrm{on}}\otimes\bm{I}_N+\sum\limits_{j=1}^{d_f}\bm{B}_j\otimes k_j(\bm{Z},\bm{Z}) \right)^{-1}\left(\bm{t}-\hat{\bm{f}}(\bm{Z})\right)\\
	\bm{\Sigma}_{\bm{a}_if_i}(\bm{z})&=k_i(\bm{z},\bm{z})\bm{B}_i\!-\!\left( \bm{B}_i\!\otimes\! k_i^T(\bm{Z},\bm{z}) \!\right) \!\! \left( \!\!
	\bm{\Sigma}_{\mathrm{on}}\!\otimes\!\bm{I}_N\!+\!\!\sum\limits_{j=1}^{d_f}\!\!\bm{B}_j\!\otimes\! k_j(\bm{Z},\bm{Z}) \!\!\right)^{\!\!\!-1}\!\!\!\!\!\left( \bm{B}_i \!\otimes\! k_i(\bm{Z},\bm{z})\!\right)
\end{align}
This follows from a trivial extension of the results in \citep{Duvenaud2014} to multiple-output
GPs. Due to the definition of $\bm{B}_i$, we can equivalently write
\begin{align}
	\mu_{\bm{a}_i f_i}(\bm{z})&=\bm{a}_i\left(\left(\bm{a}_i^T \otimes k_i^T(\bm{Z},\bm{z}) \right)\!\! \left( \!\!
	\bm{\Sigma}_{\mathrm{on}}\!\otimes\!\bm{I}_N\!+\!\!\sum\limits_{j=1}^{d_f}\!\!\bm{a}_j\bm{a}_j^T\!\otimes\! k_j(\bm{Z},\bm{Z}) \!\!\right)^{\!\!\!-1}\!\!\!\!\!\left(\bm{t}-\hat{\bm{f}}(\bm{Z})\right)\right)\\
	\bm{\Sigma}_{\bm{a}_if_i}(\bm{z})&=\bm{a}_ik_i(\bm{z},\bm{z})\bm{a}_m^T\\
	&-\bm{a}_i\!\!\left(\!\!\left(\bm{a}_i^T  \!\otimes\! k_i^T(\bm{Z},\bm{z})\right)  \!\! \left( \!\!
	\bm{\Sigma}_{\mathrm{on}}\!\otimes\!\bm{I}_N\!+\!\!\sum\limits_{j=1}^{d_f}\!\!\bm{a}_j\bm{a}_j^T\!\otimes\! k_j(\bm{Z},\bm{Z}) \!\!\right)^{\!\!\!-1}\!\!\!\!\!\left(\bm{a}_i  \!\otimes \! k_i(\bm{Z},\bm{z})\right)\!\!\right)\!\!\bm{a}_i^T,
\end{align}
from which we can directly deduce the identities \eqref{eq:muf} and \eqref{eq:sigmaf}.
{\jmlrQED}

\par\noindent{\bfseries\upshape Proof of Theorem~\ref{thm2}\ }
The result follows from Lemma~\ref{lem1} and a straightforward adaption of 
\cite[Theorem 3.1]{Lederer2019}.
{\jmlrQED}

\section{Variance Bounds and Lyapunov-Based Data Densities}

\begin{lemma}
	\label{lem5}
	The posterior variance $\tilde{\sigma}_i^2(\bm{x})$ defined in \eqref{eq:sigmaf} is bounded by
	\begin{align}
	\label{eq:varbound}
	\tilde{\sigma}_j^2(\bm{x})\leq s_{f_j}^2-\frac{s_{f_j}^4\exp(-\tilde{\phi}_j^2(\bm{x})) \sum\limits_{i=1}^{d_x}a_{i,j}^2}{\max\limits_{m=1,\ldots,d_x}\sum\limits_{n=1}^{{d_f}}\sum\limits_{i=1}^{d_x}a_{m,n}a_{i,n}s_{f_n}^2+\frac{\lambda_{\max}(\Sigma_{\mathrm{on}})}{M}}.
	\end{align}
\end{lemma}
\begin{proof}
	This result is a direct extension of \cite[Corollary 3.1]{Lederer2019a} to multiple-output GPs with 
	linear coregionalization and we pursue the proof analogously. Since the posterior variance
	is non-increasing, we can consider only training samples $\bm{z}^{(n)}$ within distance at most 
	$\tilde{\phi}_j(\bm{x})$ to $\begin{bmatrix}
	\bm{x}^T&\bm{\pi}^T(\bm{x})
	\end{bmatrix}^T$ in the posterior variance calculation. Therefore, we obtain
	\begin{align}
	\label{eq:fillbound}
		\sigma_j^2(\bm{x})\leq s_{f_j}^2-\frac{\| k_j^T(\bm{Z}_{\tilde{\phi}_j(\bm{x})},\bm{z})  \otimes\bm{a}_j^T \|^2}{\lambda_{\max}\left( \!\!
			\bm{\Sigma}_{\mathrm{on}}\!\otimes\!\bm{I}_N\!+\!\!\sum\limits_{n=1}^{d_f}\!\!\bm{a}_n\bm{a}_n^T\!\otimes\! k_n(\bm{Z}_{\tilde{\phi}_j(\bm{x})},\bm{Z}_{\tilde{\phi}_j(\bm{x})}) \!\!\right)},
	\end{align}
	where $\bm{Z}_{\tilde{\phi}_j(\bm{x})}$ denotes the training samples with distance at most 
	$\tilde{\phi}_j(\bm{x})$ to $\begin{bmatrix}
	\bm{x}^T&\bm{\pi}^T(\bm{x})
	\end{bmatrix}^T$.
	We trivially obtain the bound 
	\begin{align}
	\label{eq:normk}
		\| k_j(\bm{Z}_{\tilde{\phi}_j(\bm{x})},\bm{z})  \otimes\bm{a}_j \|\geq M s_{f_j}^4\exp(-\tilde{\phi}_j^2(\bm{x})) \sum\limits_{i=1}^{d_x}a_{i,j}^2
	\end{align}
	due to the distance restriction. Moreover, the application of Gershgorin's theorem yields	
	\begin{align}
	\label{eq:eigmax}
		\lambda_{\max}\!\!\left( \!\!
		\bm{\Sigma}_{\mathrm{on}}\!\otimes\!\bm{I}_N\!+\!\!\sum\limits_{n=1}^{d_f}\!\!\bm{a}_n\bm{a}_n^T\!\otimes\! k_n(\bm{Z}_{\tilde{\phi}_j(\bm{x})},\bm{Z}_{\tilde{\phi}_j(\bm{x})}) \!\!\right)\leq \lambda_{\max}(\bm{\Sigma}_{\mathrm{on}})\!+\!\!
		\max\limits_{m=1,\ldots,d_x}\sum\limits_{n=1}^{{d_f}}\sum\limits_{i=1}^{d_x}a_{m,n}a_{i,n}M s_{f_n}^2
	\end{align}
	due to the definition of the $M$-fill distance $\tilde{\phi}_j(\bm{x})$ in \cref{def:filldist}. Substituting the bounds \eqref{eq:normk} and \eqref{eq:eigmax} in \eqref{eq:fillbound} finally yields 
	the result.
\end{proof}

\par\noindent{\bfseries\upshape Proof of Theorem~\ref{thm6}\ }
Since $\sqrt{\beta(\tau)}\tilde{\sigma}_j(\bm{x})>\gamma_j(\tau)$ by assumption, we can simplify
\begin{align}
	\sqrt{\beta(\tau)}\tilde{\sigma}_j(\bm{x})+\gamma_j(\tau)\leq 2\sqrt{\beta(\tau)}\tilde{\sigma}_j(\bm{x}).
\end{align}
Therefore, satisfaction of the condition
\begin{align}
	4\left(\sum\limits_{i=1}^{d_x} |a_{i,j}\frac{\partial}{\partial x_i}V(\bm{x},t)|\right)^2\beta(\tau)\tilde{\sigma}_j^2(\bm{x})\leq \xi^2_{j}(\bm{x},t)
\end{align}
implies the statement of \cref{thm6}. Hence, we can substitute \eqref{eq:varbound} and solve 
for $\tilde{\phi}_j^2(\bm{x})$ in order to prove \cref{thm6}.
{\jmlrQED}

\end{document}

%% file: figure/time_var.tex
\pgfplotsset{select coords between index/.style 2 args={
    x filter/.code={
        \ifnum\coordindex<#1\def\pgfmathresult{}\fi
        \ifnum\coordindex>#2\def\pgfmathresult{}\fi
    }
}}

\pgfdeclarelayer{background}
\pgfdeclarelayer{foreground}
\pgfsetlayers{background,main,foreground}

\begin{tikzpicture}
\begin{axis}[
name=plot1,
  ylabel={state $x_2$},
  legend pos=north west,
  width=0.35\columnwidth,
  height=0.35\columnwidth,
  ymin=-1.2,
  ymax=1.2,
  xmin=-1.3,
  xmax=1.3,
    font={\sffamily},
  legend style={font=\scriptsize\sffamily,at={(2.05,0.175)},anchor=south west,/tikz/every even column/.append style={column sep=0.1cm}},
  legend columns=1,
  legend cell align={left}]

\addplot[only marks,color=green!50!black,mark=*,forget plot] table [x index=10,y index=11]{data/traj_int5.txt};
\addplot[only marks,color=black,mark=+, on background layer,forget plot] table [x index=6,y index=7]{data/traj_int5.txt};
 
\node[anchor=west] at (axis cs: -1.35,0.87) {\footnotesize\sffamily $t=\frac{4\pi}{5}$};
 
\addlegendimage{color=red,line width=1pt};
\addlegendimage{color=red,dashed,line width=1pt};
\addlegendimage{color=blue, line width=1pt};
\addlegendimage{only marks, color=blue, line width=1pt,mark=o};
\addlegendimage{only marks,color=black,mark=+};
\addlegendimage{only marks,color=green!50!black,mark=*,};
\legend{reference, future reference, trajectory, current state, all data points, selected data}
\end{axis}
\begin{axis}[
 name=plot1,
 legend pos=north west,
 width=0.35\columnwidth,
 height=0.35\columnwidth,
 ymin=-1.2,
 ymax=1.2,
 xmin=-1.3,
 xmax=1.3,
 font={\sffamily},
 hide x axis,
 hide y axis]
 
 \addplot[color=red,line width=1pt,forget plot] table [x index=0,y index=1]{data/traj_int5.txt};
 \addplot[color=red,dashed,line width=1pt,forget plot] table [x index=8,y index=9]{data/traj_int5.txt};
 \addplot[color=blue, line width=1pt,forget plot] table [x index=4,y index=5]{data/traj_int5.txt};
 \addplot[only marks, color=blue, line width=1pt,mark=o,select coords between index={199}{199},forget plot] table [x index=4,y index=5]{data/traj_int5.txt};
 
\end{axis}
\begin{axis}[
  name=plot2,
  at=(plot1.east), anchor=west,
  xshift=0.1cm,
  xlabel={state $x_1$},xlabel style={at={(-0.0,-0.07)}},
  legend pos=north west,
  width=0.35\columnwidth,
  height=0.35\columnwidth,
  ymin=-1.2,
  ymax=1.2,
  xmin=-1.3,
  xmax=1.3,
    font={\sffamily},
  yticklabels={,,}]
\addplot[only marks,color=black,mark=+, on background layer] table [x index=6,y index=7]{data/traj_int9.txt};
 \addplot[only marks,color=green!50!black,mark=*] table [x index=10,y index=11]{data/traj_int9.txt};
 \node[anchor=west] at (axis cs: -1.35,0.87) {\footnotesize\sffamily $t=\frac{8\pi}{5}$};
 \end{axis}
 \begin{axis}[
 name=plot2,
 at=(plot1.east), anchor=west,
 xshift=0.1cm,
 legend pos=north west,
 width=0.35\columnwidth,
 height=0.35\columnwidth,
 ymin=-1.2,
 ymax=1.2,
 xmin=-1.3,
 xmax=1.3,
 font={\sffamily},
 yticklabels={,,},
 hide x axis,
 hide y axis]
 
 \addplot[color=red,line width=1pt] table [x index=0,y index=1]{data/traj_int9.txt};
 \addplot[color=red,dashed,line width=1pt] table [x index=8,y index=9]{data/traj_int9.txt};
 \addplot[color=blue, line width=1pt] table [x index=4,y index=5]{data/traj_int9.txt};
 \addplot[only marks, color=blue, line width=1pt,mark=o,select coords between index={199}{199}] table [x index=4,y index=5]{data/traj_int9.txt};
\end{axis}
\end{tikzpicture}